 \documentclass[10pt,reqno]{amsart}
 \usepackage[a4paper,top=3cm, bottom=3cm,left=3cm, right=3cm, twoside]{geometry}
  \usepackage[foot]{amsaddr}

  \usepackage{amsthm}
  \usepackage{amsmath}
  \usepackage{amsfonts}
  \usepackage{amssymb}
  \usepackage{amscd}
  \usepackage{url}
 \usepackage{bbm} 
 \usepackage{bm}
 \usepackage{bold-extra}
 \usepackage{textcomp}
 \usepackage{graphicx}
  \usepackage[mathscr]{eucal}
  \usepackage{mathtools} 
  \usepackage{enumitem}
  \usepackage{tikz} 
 \usepackage{tikz-cd}
 \usepackage{wasysym}
  \usepackage{physics}
  \usepackage[english]{babel}
 \usepackage[mathscr]{eucal}
 \usepackage{lmodern}
 \usepackage{bera}%
 \usepackage{fancybox}
 \usepackage[T1]{fontenc} 
 \usepackage[utf8]{inputenc}
  \usepackage{todonotes}
  \usepackage{fancyhdr}
  \usepackage{hyperref}

\newtheorem{proposition}{\bf Proposition}[section]
\newtheorem{lemma}{\bf Lemma}[section]
\newtheorem{theorem}{\bf Theorem}[section]

\theoremstyle{definition}  
   \newtheorem{defn}{Definition}[section]
   
   \newtheorem{eg}[defn]{Example}

   \newtheorem{rmk}[defn]{Remark}

  \theoremstyle{plain}  
   \newtheorem{prop}[defn]{Proposition}
   \newtheorem{cor}[defn]{Corollary}

  \theoremstyle{remark} 

\newcommand{\msc}[1]{\mathscr{#1}}
 \newcommand{\mbb}[1]{\mathbb{#1}}
   \newcommand{\B}[1]{\mathscr{B}({#1})}

   \newcommand{\CH}{\mathcal{H}}
   \newcommand{\CK}{\mathcal{K}}
    
   \newcommand{\GCk}{\Gamma(\mathbb{C}^{k})}
   
   \newcommand{\GCd}{\Gamma(\mathbb{C}^{d})}
   \newcommand{\GH}{\Gamma(\mathcal{H})}
   \newcommand{\GL}{\Gamma(L)}

     \newcommand{\BC}{\mathbb{C}}
   \newcommand{\BR}{\mathbb{R}}

    \newcommand{\TC}[1]{\msc{T}({#1})}

   \newcommand{\bu}{\mathbf{u}}
   \newcommand{\bv}{\mathbf{v}}
   
   \newcommand{\m}{\mathbf{m}}
   \newcommand{\w}{\mathbf{w}}
   \newcommand{\x}{\mathbf{x}}
   \newcommand{\y}{\mathbf{y}}
   \newcommand{\z}{\mathbf{z}}

 \newcommand{\ul}[1]{\underline{#1}}

\DeclareMathOperator{\diag}{Diag}

\DeclareMathOperator{\re}{Re}
\DeclareMathOperator{\im}{Im}

\makeatletter
\@namedef{subjclassname@2020}{\textup{2020} Mathematics Subject Classification}
\makeatother

\title[What is a Gaussian channel?]{What is a Gaussian channel, and when is it physically implementable using a multiport interferometer?}

\author{Repana Devendra}
\address{Department of Mathematics, Indian Institute of Technology Bombay,  Mumbai, Maharashtra, 400076, India}
\email{r.deva1992@gmail.com}

\author{Tiju Cherian John}
\address{Department of Electrical and Computer Engineering, The University of Arizona, Tucson, AZ, 85719, USA}
\email{tijucherian@fulbrightmail.org}

\author{K. Sumesh}
\address{Department of Mathematics, Indian Institute of Technology Madras,  Chennai, Tamil Nadu, 600036, India}
\email{sumeshkpl@gmail.com}

\date{\today}

\begin{document}

\maketitle

\begin{abstract}
Quantum Gaussian channels are fundamental models for communication and information processing in continuous-variable quantum systems. This work addresses both foundational aspects and physical implementation pathways for these channels. Firstly, we provide a rigorous, unified framework by formally proving the equivalence of three principal definitions of quantum Gaussian channels prevalent in the literature,  consolidating theoretical understanding. Secondly, we investigate the physical realization of these channels using multiport interferometers, a key platform in quantum optics. We focus on the standard parameterization involving matrices $X$ and $Y$, governed by the condition $Y - i(J - X^T J X) \ge 0$ with $J=\begin{bsmallmatrix} 0 & I \\ -I & 0 \end{bsmallmatrix}$. The central research contribution is a precise characterization of the specific pairs $(X, Y)$ that correspond to Gaussian channels physically implementable via linear optical multiport interferometers. This characterization bridges the abstract mathematical description with concrete physical architectures. Along the way, we also resolve some questions posed by Parthasarathy [Indian J. Pure Appl. Math. 46, 419–439 (2015)].
 

\end{abstract}



\tableofcontents

\section{Introduction}
 Quantum channels are fundamental for classical and quantum communication systems, describing the dynamics of quantum states in various physical processes, including noise, decoherence, and information transfer \cite{Watrous_2018, Wilde_2013, Nielsen_Chuang_2010}. 
 Among quantum  channels,  Gaussian channels  are particularly important in the study of continuous-variable  quantum systems. This article focuses on this particular class of quantum channels. They play a fundamental role in areas such as open quantum systems, classical-quantum communication systems like optical fibers, general quantum communication, quantum error correction, and quantum computation, making their theoretical and practical understanding crucial \cite{Serafini2017-mz, Sabapathy2017-xp, Weedbrook2012-zl, Guha2008-iw, Holevo2001-zc}. The mathematical characterization of Gaussian channels is a rich and well-studied area, with multiple definitions appearing in the literature. These definitions include descriptions based on the action on transformation of Gaussian states (or equivalently, mean and covariance matrices), Weyl unitary operators, and Gaussian unitary dilation \cite{Holevo2019-zo, Serafini2017-mz,Eisert2007-rl, Holevo2001-zc, Fannes1976-il}. Although it is implicitly understood that these definitions are all equivalent, it has not been formally proved in the literature as a unified theorem. This gap in formalization not only leaves the equivalence as an unstated assumption, but also poses challenges for researchers seeking a mathematically rigorous foundation for their work.

 Beyond their theoretical significance, Gaussian channels are crucial for the experimental realization of quantum technologies. A key practical question is their implementability using optical devices such as multiport interferometers, which perform linear transformations of optical modes through passive components such as beam splitters and phase shifters \cite{Reck1994-un}. Although multiport interferometers are widely used in quantum optics, a detailed mathematical framework to determine the conditions under which they can implement a given Gaussian channel remains absent from the literature. This article aims to address these issues with a mathematically oriented treatment that unifies existing knowledge and extends it to new domains.

 The article is organized as follows: In Section \ref{sec:prelim}, we introduce the basic objects of the theory. In Section \ref{sec:channels}  we establish the equivalence of various definitions of the Gaussian channel through rigorous mathematical arguments. The main technical result that aids our proof is Proposition \ref{prop-symplectic-extension} which is inspired from \cite[page 99]{Serafini2017-mz}. 
 Section \ref{sec:optics}  presents the characterization of the physical implementability of Gaussian channels using multiport interferometers. We investigate the parameterization of Gaussian channels using two matrices $X$ and $Y$ satisfying the condition $Y-i(J-X^TJX)\geq 0$, where $J=\bmqty{0&I\\-I&0}$, focusing on the conditions under which they can be physically implemented using multiport interferometers.  
 Along the way, we also answer affirmatively some questions asked by Parthasarathy in \cite{KRP15}.

\section{Preliminaries}\label{sec:prelim}

    For any Hilbert space $\CH$, whether it is finite-dimensional or infinite-dimensional, and regardless of whether it is defined over the field of real or complex numbers, we denote the space of all bounded linear operators on $\CH$ as $\B{\CH}$. The adjoint of a real linear operator $A$ defined in a real Hilbert space is denoted by $A^T$ and that of a complex linear operator $A$ defined in a complex Hilbert space is denoted by $A^\dagger$.  In the space of all bounded self-adjoint operators ($A=A^{\dagger}$), as well as in the space of real symmetric ($A=A^{T}$) or complex Hermitian matrices, we define a partial order, denoted by $A\leq B$. This ordering is defined by the condition that $\mel{\z}{(B-A)}{\z}\geq 0$ for all vectors $\z\in\CH$. If $0\leq B$, then we say that $B$ is a \textit{positive operator} (or \textit{matrix}, depending on the context). We  say that $A$ is a \textit{strictly positive operator} (or \textit{matrix}) if it is positive and invertible. The symbol $\geq $ is  used with obvious meanings arising from the discussion above. 
    
    The space of all trace-class operators on $\CH$ is denoted by $\TC{\CH}$. By a \textit{state} in $\CH$ we mean a positive trace-class operator $\rho\in\TC{\CH}$ satisfying $\tr(\rho) = 1$. 
    If $\CH$ and $\CK$ are Hilbert spaces and $\rho$ is a state on $\CH\otimes \CK$, then $\tr_2(\rho)\in\TC{\CH}$ denote the partial trace of $\rho$ over  the second factor $\CK$.
    For $d_1,d_2\in\mathbb{N}$, we define $M_{d_1\times d_2}(\mbb{C})$ and $M_{d_1\times d_2}(\mbb{R})$ to be the space of all $d_1\times d_2$ complex matrices and real matrices, respectively. In particular, $M_d(\mbb{C})=M_{d\times d}(\mbb{C})$ and $M_d(\mbb{R})=M_{d\times d}(\mbb{R})$.

    Any vector $\z\in \BC^d$ (or $\BR^d$) is represented as  a column vector $\z = [z_1,z_2,\dots,z_d]^T$ in the standard orthonormal basis of $\BC^d$ (or $\BR^d$). The  round bracket  $(z_1,z_2,\dots,z_d)$ is often used to refer to the column vector $\z$. For a vector $\z = (z_1,z_2,\dots,z_d)\in\BC^d$ we define $\bar{\z} := (\bar{z}_1,\bar{z}_2,\dots,\bar{z}_d)$, where $\bar{z}_j$ denotes the complex conjugate of $z_j\in\BC$. Given $\z = (z_1,z_2,\dots,z_d)$ and $\w =(w_1,w_2,\dots,w_d)$ in $\BC^d$ we define
    \begin{align*}
       \z^T\w := \sum\limits_{j=1}^d z_jw_j; \qquad
        \braket{\z}{\w}:=\bar{\z}^T\w= \sum\limits_{j=1}^d\bar{z}_jw_j; \qquad
        \norm{\z} :=\braket{\z}{\z}^{1/2}.
    \end{align*}
   We consider $\BC^d$ as a $2d$-dimensional  real Hilbert space with $\re \braket{\cdot}{\cdot}$ as the inner product, and as a \textit{symplectic space} with its standard symplectic form  $-i\im\braket{\cdot}{\cdot}$. A set $\{\z_1,\cdots,\z_k,\w_1,\cdots,\w_k\}$ $\subseteq~\mbb{C}^{d},1\leq k\leq d$ is said to be a \textit{symplectic set} if 
    \begin{align*}
        \Im\braket{\z_i}{\z_j}=0=\Im\braket{\w_i}{\w_j} \quad\text{ and }\quad \Im\braket{\z_i}{\w_j}=\delta_{ij},\qquad\forall~1\leq i,j\leq k.
    \end{align*}
    If $k=d$, then a symplectic set will form a basis for the real vector space $\mbb{C}^d$, and in such a case, we call it a \textit{symplectic basis}. It is well known that any symplectic set can be extended to a symplectic basis \cite[Theorem 1.15]{Gosson2006-lr}. 
    
    We identify $\BC^d$ with $\BR^{2d}$ as real Hilbert spaces in two different ways. With these identifications, we also discuss the symplectic bases in $\BR^{2d}$. In the first case, we identify $\BC^d=\BR^d+i\BR^d$ with $\BR^{2d}=\BR^d\oplus\BR^d$ via the identification $\z=\x+i\y\mapsto (\x,\y)$. Under this identification, we realize an element $A+iB$ of $M_{d}(\BC)= M_d(\BR)+iM_d(\BR)$ as the element $\bmqty{A&-B\\B&A}$ of $M_{2d}(\BR)=M_2(M_d(\BR))$. In particular, $-iI_d\in M_d(\BC)$ corresponds to the block matrix 
    \begin{align*}
        J_{2d}:=\bmqty{0&I_d\\-I_d&0}\in M_{2d}(\BR),
    \end{align*}
    where $I_d$ is the identity matrix of size $d\times d$. Note that $J_{2d}$ is invertible with the inverse $J_{2d}^{-1}=J_{2d}^T=-J_{2d}$. In the second case, we write $d=d_1+d_2$, and write $\BC^d=\BC^{d_1}\oplus\BC^{d_2}=\BR^{2d_1}\oplus\BR^{2d_2}$ via the identification
    \begin{align*}
        (\x_1+i\y_1) \oplus (\x_2+i\y_2)\mapsto (\x_1,\y_1, \x_2,\y_2),\qquad\forall~\x_j,\y_j\in\BR^{d_j}, j=1,2.
    \end{align*}
    Under this identification, given $A_{i,j},B_{i,j}\in M_{d_i\times d_j}(\BR)$  we identify 
    \begin{align*}
        M_d(\BC)\ni\bmqty{A_{11}+iB_{11}&A_{12}+iB_{12}\\ A_{21}+iB_{21}&A_{22}+iB_{22}}
        \mapsto
        \bmqty{A_{11}&-B_{11}&A_{12}&-B_{12}\\ B_{11}&A_{11}&B_{12}&A_{12} \\ A_{21}&-B_{21}&A_{22}&-B_{22}\\ B_{21}&A_{21}&B_{22}&A_{22}}\in M_{2d}(\BR).
    \end{align*}
    In particular, $-iI_d\in M_d(\BC)$ corresponds to the block matrix $J_{2d_1}\oplus J_{2d_2}$. Note that, under both identifications, the adjoint of a complex matrix corresponds to the transpose of its identification.


\subsection{Symplectic and Orthosymplectic Matrices}\label{sec:symp}

By a \textit{symplectic transformation} of $\BC^d$ we mean a real linear operator $L$ on $\BC^d$ that preserves the standard symplectic form, i.e.,
 \begin{align}\label{eq-Sym-Tran}
     \im \braket{L\bu}{L\bv} = \Im\braket{\bu}{\bv}\qquad\forall~\bu, \bv \in \BC^d.
 \end{align}
 Now, fix an identification of $\BC^d$ with $\BR^{2d}$, as described in the previous section. This identification allows us to represent the real linear transformation $L$ as a $2d\times 2d$ real matrix (again denoted as $L$). Consider a specific $J\in \{J_{2d}, J_{2d_1} \oplus J_{2d_2}\}$, where $d=d_1+d_2$. It is important to note that $\z^TJ\z'=\Im\braket{\z}{\z'}$ for all $\z,\z'\in\BC^d=\BR^{2d}$. Therefore, the equation \eqref{eq-Sym-Tran} is equivalent to stating that $L^TJL=J$.  The elements of the set 
    \begin{align*}
        Sp(2d, J):=\{L\in M_{2d}(\BR): L^TJL=J\}
    \end{align*}
    are called \textit{$J$-symplectic} matrices. A $J$-symplectic matrix that is also orthogonal is called a \textit{$J$-orthosymplectic} matrix. The set $Sp(2d, J)$ forms a group under multiplication, called \textit{$J$-symplectic group}. Clearly, $J\in Sp(2d, J)$. Every $J$-symplectic matrix $L$ has determinant $1$, and its inverse is given by $L^{-1}=J^TL^TJ$. Note also that $L\in Sp(2d, J)$ implies  $L^T\in Sp(2d, J)$. 
 A $J_{2d}$-symplectic matrix and a $J_{2d}$-orthosymplectic matrix are simply called a \textit{symplectic} matrix and an \textit{orthosymplectic} matrix, respectively.  Note that if  $L=\bmqty{\bu_1&\bu_2&\cdots&\bu_{2d}}\in M_{2d}(\BR)$ and $J=[J_{ik}]\in M_{2d}(\BR)$, then   
     \begin{align}\label{eq-sym-matrix-basis}
         L\in Sp(2d, J)
             &\Longleftrightarrow \bu_i^TJ\bu_k=J_{ik},\qquad\forall~1\leq i,k\leq 2d. 
     \end{align}
     Consequently, a matrix $L=\bmqty{\bu_1,\cdots,\bu_d,\bv_1,\cdots,\bv_{d}}\in M_{2d}(\BR)$ is a symplectic if and only if its columns form a symplectic basis.

\subsection{Weyl Representation and Bogoliubov Transformations}

 The \textit{symmetric Fock space} over $\BC^d$ is defined by
 \begin{align*}
        \Gamma(\BC^d):= \bigoplus_{k=0}^{\infty}(\BC^d)^{{\text{\textcircled{s}}}^k},
 \end{align*}
 where $(\BC^d)^{{\text{\textcircled{s}}}^k}$ is the $k$-times \textit{symmetric tensor product} of $\BC^d$ (see \cite[Chapter 2]{Krp92} for details). For $\bv\in \BC^d$, the \textit{exponential vector} $\ket{e(\bv)}$ is defined as 
 \begin{equation*}
    \ket{e(\bv)} := \bigoplus_{k=0}^\infty  
    \frac{\bv^{\otimes^k}}{\sqrt{k!}},
 \end{equation*}
 where $\bv^{\otimes^0} := 1\in\BC$. The span of exponential vectors forms a dense subset of $\Gamma(\BC^d)$. 
 For $\bu\in \BC^d$, the \textit{Weyl unitary operator}, $W(\bu)$ is the unique unitary operator on $\GCd$ satisfying  
 \begin{equation*}
    W(\bu) \ket{e(\bv)}= \exp{-\frac{1}{2}\|\bu\|^2-\braket{\bu}{\bv}} \ket{e(\bu+\bv)}, \qquad \forall~\bv\in \BC^d. 
 \end{equation*}
 The adjoint of $W(\bu)$ is given by  $W(\bu)^{\dagger}=W(-\bu)$.  The Weyl unitary operators satisfy the Weyl form of the the \textit{canonical commutation relations}: 
 \begin{align*}
      W(\bu)W(\bv) &= \exp\{-i \Im \braket{\bu}{\bv}\}W(\bu+\bv),  \qquad \forall~ \bu, \bv \in \BC^d; \\
      W(\bu)W(\bv) &= \exp\{-2i \Im \braket{\bu}{\bv}\}W(\bv)W(\bu),  \qquad \forall~ \bu, \bv \in \BC^d.
 \end{align*}
The map $\bu\mapsto W(\bu)$, $\bu\in \BC^d$ is called the \emph{Weyl representation} of the additive group $\BC^d$. Now
write $d=d_1+d_2$, where $d_1,d_2\in\mathbb{N}$. Then  $\BC^d= \BC^{d_1}\oplus \BC^{d_2}$, and the map given by   $\ket{e(\bv_1\oplus\bv_2)}\mapsto \ket{e(\bv_1)}\otimes \ket{e(\bv_2)},$ for all $\bv_1\in \BC^{d_1}, \bv_2\in \BC^{d_2}$ extends to a Hilbert space isomorphism between $\Gamma(\BC^d)$ and $\Gamma(\BC^{d_1})\otimes \Gamma(\BC^{d_2})$. Under this isomorphism, we identify these two Fock spaces and write $\Gamma(\BC^d)= \Gamma(\BC^{d_1})\otimes \Gamma(\BC^{d_2})$
 and we have $W(\bu_1\oplus\bu_2)=   W(\bu_1)\otimes W(\bu_2),$ for all $\bu_1\in\BC^{d_1},\bu_2\in\BC^{d_2}$
 called the \textit{factorizability} of the Weyl representation. Finally, the map $\bu\mapsto W(\bu)$, $\bu\in \BC^d$ is a  strongly continuous, projective, factorizable and irreducible unitary representation of the additive abelian group $\BC^d$.

 Another important class of unitary operators we need in this article consists of  the so called \textit{Bogoliubov transformations} on $\GCd$. 
  Given a symplectic transformation $L$ on $\BC^d$,
 there exists a unique unitary  $\GL$ on $\GCd$, called the \textit{Bogoliubov unitary} at $L$,  satisfying  $\GL W(\bu) \GL^{\dagger} = W(L\bu),$
  and $\mel{e(0)}{\GL}{e(0)}>0.$ In this case, we have $\Gamma(L)^{\dagger}=\Gamma(L^{-1})$. We refer the reader to  \cite[Chapter 2]{Krp92} and \cite[Section IV]{TiKR21} for more details on the topics  mentioned in this sub section.

\subsection{Gaussian States}
In this section, we define Gaussian states  and describe their properties which we need in a basis independent way. This  will help us provide rigorous  proofs later.
\begin{defn}
    Let $\rho$ be a state on $\GCd$. Then the map $\widehat{\rho}:\BC^{d}\rightarrow \BC$ defined by  
    \begin{equation}\label{eq:10.1}
    \end{equation}
    is called the \textit{quantum characteristic function} of $\rho$. If there exist a vector $\m \in \BC^d$ and a real linear operator $S$ on $\BC^{d}$ such that 
    \begin{equation}
     \widehat{\rho}(\z) =  \exp{-i\re\braket{\m}{\z}- \frac{1}{2}\re \mel{\z}{S}{\z}}, \qquad \forall~ \z\in \BC^d,
    \end{equation} 
    then $\rho$ is called a \textit{$d$-mode quantum Gaussian state} or simply a \textit{Gaussian state}. In this case, the pair $(\m,S)$ is unique;  the vector $\m$ is called the mean vector and $S$ is called the covariance operator of the Gaussian state $\rho$; and we  write $\rho = \rho_{(\m,S)}$.
\end{defn}

 Let $\rho$ be a Gaussian state with  mean vector $\m$ and covariance operator $S$.
     If $V$ denotes an identification of $\BC^d$ with $\BR^{2d}$ as described in Section \ref{sec:prelim},  then the quantum characteristic function of  $\rho$ can be considered as  the map $\widehat{\rho}:\BR^{2d}\rightarrow \BC$ given by  
    \begin{equation}\label{eq-GS-char-fun-real}
     \widehat{\rho}(V\z) =  \exp{-i(V\m)^T(V\z)- \frac{1}{2}(V\z)^T(VSV^T)(V\z)}, \qquad \forall~ \z \in \BC^d,
    \end{equation} 
    where $VSV^T\in M_{2d}(\BR)$. When the identification is clear from the context, we abuse the notation and write  $\m\in \BR^{2d}$ and $S\in M_{2d}(\BR)$. Furthermore, an arbitrary matrix $S\in M_{2d}(\BR)$ is the covariance matrix of a Gaussian state if and only if it belongs to the set 
    \begin{equation*}
     \msc{CM}(d):=\{A\in M_{2d}(\BR): A+iJ\geq 0\},
    \end{equation*} where $J$ is taken as $J_{2d}$ or $J_{2d_1}\oplus J_{2d_2}$ according to the identification we have chosen. For more details on quantum Gaussian states, we refer to \cite{Parthasarathy2010-su}, and for the basis-independent description of Gaussian states, see \cite{john2018infinite}. The following proposition is a consequence of the definitions.
 
 \begin{prop}\label{sec:quant-gauss-stat-1}
    Let $\rho = \rho_{(\m, S)}$ be a Gaussian state on $\GCd$. 
    If $\bu\in \BC^{d}$ and $L$ is a symplectic transformation of $\BC^d$, then 
    \begin{align*}
      W(\bu)\GL\rho_{(\m, S)}\GL^\dagger W(\bu)^\dagger = \rho_{((L^{-1})^T\m-i2\bu, (L^{-1})^TSL^{-1})}.
    \end{align*}
    In particular, if $L$ is a unitary operator then 
    \begin{equation*}
         W(\bu)\Gamma(L)\rho_{(\m, S)} \Gamma(L)^{\dagger}W(\bu)^{\dagger}= \rho_{(L\m-i2\bu, LSL^T)}.
    \end{equation*}
 \end{prop}
 
 \begin{rmk}\label{rmk:quant-gauss-stat-1}
    Let $\rho = \rho_{(\m, S)}$ be a Gaussian state on $\GCd$. 
    If $\bu\in \BR^{2d}$ and $L\in Sp(2d,J)$, then 
    \begin{align*}
      W(\bu)\GL\rho_{(\m, S)}\GL^\dagger W(\bu)^\dagger = \rho_{((L^{-1})^T\m+2J\bu, (L^{-1})^TSL^{-1})}.
    \end{align*}
    In particular, if $L$ is a $J$-orthosymplectic matrix then 
    \begin{equation*}
         W(\bu)\Gamma(L)\rho_{(\m, S)} \Gamma(L)^{\dagger}W(\bu)^{\dagger}= \rho_{(L\m+2J\bu, LSL^T)}.
    \end{equation*}
 \end{rmk}
 
\begin{defn}
    A unitary operator $U$ on $\GCd$ is called a \textit{Gaussian unitary operator} (or \textit{Gaussian symmetry}) if $U\rho U^{\dagger}$ is a Gaussian state for every Gaussian state $\rho$ on $\GCd$.
\end{defn}


\begin{prop}[\cite{KRP12, john2018infinite}] \label{prop:gaussian-unitary}
    Let $U$ be a unitary operator on $\GCd$. Then $U$ is a Gaussian unitary operator if and only if $ U=\lambda W(\bu)\Gamma(L)$
    for some $\lambda \in \BC$ with $\abs{\lambda}= 1$, $\bu\in \BC^{d}$ and some symplectic transformation $L$ on $\BC^{d}$. 
\end{prop}

\section{Gaussian Channels}\label{sec:channels}

 This section discusses quantum channels defined on Fock spaces. Mathematically, quantum channels are represented by trace-preserving completely positive maps (in the Schrödinger picture), or equivalently, unital completely positive maps (in the Heisenberg picture). To ensure clarity, we will provide a brief overview of each term used in this context.
 
 A linear map $\Phi:\B{\GCd}\to\B{\GCd}$ is said to be \textit{completely positive} (CP) if the positivity condition $\sum_{j,k=1}^nB_j^*\Phi(A_j^*A_k)B_k\geq 0$ is satisfied for any finite subsets  $\{A_j\}_{j=1}^n,\{B_j\}_{j=1}^n\subseteq\B{\GCd}$ and $n\in\mathbb{N}$. We refer to \cite{Sti55,Pau02} for more details on CP maps and its properties. Given any bounded (with respect to the trace-norm) linear map $\Psi:\TC{\GCd}\to\TC{\GCd}$ there exists a unique bounded (with respect to the operator-norm) linear map  $\Psi^*:\B{\GCd}\to\B{\GCd}$, called the \textit{dual} of $\Psi$,  such that $\tr(\Psi(\varrho)A)=\tr(\varrho\Psi^*(A))$ for all $\varrho\in\TC{\GCd}$ and $A\in\B{\GCd}$. The map $\Psi^*$ is necessarily a \textit{normal} map, i.e., continuous with respect to the $\sigma$-weak (or weak$^*$) topology on $\B{\GCd}$.  The map $\Psi$ is said to be \textit{completely positive} if $\Psi^*$ is a completely positive map. 
 Furthermore, $\Psi$ is \textit{trace-preserving} (i.e., $\tr(\Psi(\varrho))=\tr(\varrho)$ for all $\varrho\in\TC{\GCd}$) if and only if $\Psi^*$ is \textit{unital}  (i.e., $\Psi^*(I)=I$, where $I$ is the identity map of $\GCd$). A trace-preserving completely positive map  $\Psi:\TC{\GCd}\to\TC{\GCd}$ is called a \textit{quantum channel} on $\GCd$. According to the Stinespring representation theorem \cite{Sti55, AttQ13}, if $\Psi$ is a quantum channel on $\GCd$, then there exists a separable Hilbert space $\mathcal{K}$, a unitary operator $U$ on $\BC^d\otimes\mathcal{K}$, and a state $\rho_{\CK}$ on $\mathcal{K}$ such that the channel acts on an operator $\varrho\in\TC{\GCd}$ as follows: 
 \begin{align*}
     \Psi(\varrho)=\tr_2(U(\varrho\otimes\rho_{\CK})U^\dagger).
 \end{align*}
 In this section, we aim to characterize all quantum channels that map Gaussian states to Gaussian states. 

\begin{proposition}\label{prop-mean-covar-gaussian symmetry}
     Let $d=d_1+d_2$, $\bu=\bu_1\oplus \bu_2\in\mbb{R}^{2d_1}\oplus\mbb{R}^{2d_2}$ and $L\in Sp(2d, J)$  with $L^{-1}=\bmqty{L_{11}&L_{12}\\L_{21}&L_{22}}$, $L_{jj}\in M_{2d_j}(\BR)$, $j=1,2$.  Then the map $\Psi:\TC{\Gamma(\mbb{C}^{d_1})}\to \TC{\Gamma(\mbb{C}^{d_1})}$ defined by 
    \begin{align}\label{eq-GQC-stinespring}
        \Psi(\varrho):=\tr_{2}\left(W(\bu)\Gamma(L)(\varrho\otimes \rho_{(0,I_{d_2})})(W(\bu)\Gamma(L))^{\dagger}\right), \qquad \forall~ \varrho\in\TC{\Gamma_s(\mbb{C}^{d_1})},
    \end{align}
    is a quantum channel that maps Gaussian states to Gaussian states. Indeed,  
    \begin{align}\label{eq-GQC-stinespring-gaussian}
        \Psi(\rho_{(\m,S)})=\rho_{(L_{11}^{T}\m+2J_{2d_1}\bu_1, L_{11}^{T}SL_{11}+L_{21}^{T}L_{21})},
    \end{align}
    for every Gaussian states $\rho_{(\m,S)}$ on $\Gamma(\mbb{C}^{d_1})$.
\end{proposition}

\begin{proof}
    Since $U:=W(\bu)\Gamma(L)$ is a unitary, by Stinespring representation theorem, we know that $\Psi$ is a quantum channel. Further, given any Gaussian state $\rho_{(\m,S)}$ the quantum characteristic function of $\Psi(\rho_{\m,S})$ at any $\bv\in \mbb{R}^{2d_1}$ is given by 
    \begin{align*}
    \widehat{\Psi(\rho_{(\m,S)})}(\bv)
          &= \tr\left(\tr_{2}\left(W(\bu)\Gamma(L)\big(\rho_{(\m,S)}\otimes \rho_{(0,I_{d_2})}\big)(W(\bu)\Gamma(L))^{\dagger}\right)W(\bv)\right)\\
          &=\tr\left(\big(\rho_{(\m,S)}\otimes \rho_{(0,I_{d_2})}\big)\Gamma(L)^{\dagger}W(-\bu)W(\bv\oplus\bm{0})W(\bu)\Gamma(L)\right)\\
          &=\tr\left(\big(\rho_{(\m,S)}\otimes \rho_{(0,I_{d_2})}\big)\Gamma(L^{-1})W(\bv\oplus \bm{0}) \Gamma(L^{-1})^{\dagger}\right)\exp{2i\Im\braket{\bu_1}{\bv}}\\
          &=\tr\left(\big(\rho_{(\m,S)}\otimes \rho_{(0,I_{d_2})}\big)W(L^{-1}(\bv\oplus \bm{0}))\right)\exp{2i\bu_1^TJ_{2d_1}\bv}\\
          &=\tr\left(\big(\rho_{(\m,S)}\otimes \rho_{(0,I_{d_2})}\big)W(L_{11}\bv\oplus L_{21} \bv)\right)\exp{2i\bu_1^TJ_{2d_1}\bv}\\
          &=\tr\left(\rho_{(\m,S)}W(L_{11}\bv)\right) \tr\left( \rho_{(0,I_{d_2})}W(L_{21}\bv)\right)\exp{2i\bu_1^TJ_{2d_1}\bv}\\
          &=\exp{-i \left(L_{11}^T\m+2J_{2d_1}\bu_1\right)^T\bv -\frac{1}{2}\bv^T\left(L_{11}^TSL_{11}+L_{21}^TL_{21}\right)\bv}.
    \end{align*}
    Since $\bv\in \mbb{R}^{2d_1}$ is arbitrary we conclude that $\Psi(\rho_{(\m,S)})$ is a Gaussian state with mean  $L_{11}^T\m+2J_{2d_1}\bu_1$ and covariance $L_{11}^TSL_{11}+L_{21}^TL_{21}$.
\end{proof}

 To achieve our main goal, we will first establish some technical results.

\begin{lemma}\label{lem-intertwing-J and D-real}
    Let $\lambda_j\in [-1,1]$ for $1\leq j\leq d$. Then there exists $Q\in M_{2d\times 4d}(\mbb{R})$ such that $QQ^T=I_{2d}$ and
    \begin{align*}
        Q\Big(\bigoplus_{j=1}^{2d}\bmqty{0&1\\-1&0}\Big)Q^T=\bigoplus_{j=1}^d\bmqty{0&\lambda_j\\-\lambda_j&0}.
    \end{align*}
\end{lemma}

\begin{proof}
    For any $\theta\in\mbb{R}$, let 
    \begin{align*}
        Q(\theta):=\bmqty{\cos\theta&0&-\sin\theta&0\\0&\cos\theta&0&\sin\theta}.
    \end{align*}
    Then, we have $Q(\theta)Q(\theta)^T=I_2$, and
    \begin{align*}
        Q(\theta)\Big(\bigoplus_{j=1}^2\bmqty{0&1\\-1&0}\Big)Q(\theta)^T
         &=\bmqty{0&\cos2\theta\\-\cos2\theta&0}.
    \end{align*}
    Now, for each $1\leq j\leq d$, choose and fix $\theta_j\in\mbb{R}$ be such that $\cos 2\theta_j=\lambda_j$. Then $Q:=\oplus_{j=1}^dQ(\theta_j)$ gives the required matrix. 
\end{proof}

\begin{cor}\label{cor-intertwing-J and D}
    Let $A\in M_{d}(\mbb{R})$ be a positive contraction. Then there exists a matrix $Q\in M_{2d\times 4d}(\mbb{R})$ such that $QQ^T=I_{2d}$ and $QJ_{4d}Q^T=\bmqty{0&A\\-A&0}$.
\end{cor}  

\begin{proof}
    By the spectral theorem we know that $A=UDU^T$ for some $D\in M_{2d}(\mbb{R})$ diagonal and $U\in M_{2d}(\mbb{R})$ orthogonal. Then 
    \begin{align*}
        \bmqty{0&A\\-A&0}=\bmqty{0&UDU^T\\-UDU^T&0}
                         =\bmqty{U&0\\0&U}\bmqty{0&D\\-D&0}\bmqty{U&0\\0&U}^T. 
    \end{align*}
    So it is enough to prove the result for $A$ diagonal, say $A=\diag(\lambda_1,\lambda_2,\cdots,\lambda_d)$ with $\lambda_j\in [0,1]$. Now, by the above Lemma, there exists a matrix $Q\in M_{2d\times 4d}(\mbb{R})$ such that $QQ^T=I_{2d}$ and
    \begin{align*}
        Q\Big(\bigoplus_{j=1}^{2d}\bmqty{0&1\\-1&0}\Big)Q^T=\bigoplus_{j=1}^d\bmqty{0&\lambda_j\\-\lambda_j&0}.
    \end{align*} 
    By taking conjugations with suitable permutation matrices, we obtain a matrix, again denoted by $Q\in M_{2d\times 4d}(\mbb{R})$ such that $QQ^T=I_{2d}$ and $QJ_{4d}Q^T=\bmqty{0&A\\-A&0}$.
\end{proof}

\begin{prop}\label{prop-symplectic-extension}
    Let $X,Y\in M_{2d}(\BR)$. Then the following are equivalent:
    \begin{enumerate}[label=(\roman*)]
        \item There exists  a $(J_{2d}\oplus J_{2d'})$-symplectic matrix  $L=\bmqty{X&L_{12}\\L_{21}&L_{22}}\in M_{2(d+d')}(\BR)$ with $ L_{21}^{T}L_{21}=Y$, where $L_{12}^{T},L_{21}\in M_{2d'\times 2d}(\BR), L_{22}\in M_{2d'}(\BR)$ for some $d'\geq 1$;
        \item $Y+i(J_{2d}-X^TJ_{2d}X)\geq 0$;
        \item $Y-i(J_{2d}-X^TJ_{2d}X)\geq 0$.
    \end{enumerate}
    (In each of the above statements, the matrix $Y$ is necessarily positive.)
\end{prop}
 
\begin{proof}
    $(i)\Rightarrow (ii)$ Suppose there exist $L_{12}^{T},L_{21}\in M_{2d'\times 2d}(\BR), L_{22}\in M_{2d'}(\BR)$ such that  $ L_{21}^{T}L_{21}=Y$ and the matrix $L=\bmqty{X&L_{12}\\L_{21}&L_{22}}\in M_{2(d+d')}(\BR)$ is $(J_{2d}\oplus J_{2d'})$-symplectic. Then
    \begin{align*}
     \bmqty{X&L_{12}\\L_{21}&L_{22}}^T\bmqty{J_{2d}&0\\0&J_{2d'}}\bmqty{X&L_{12}\\L_{21}&L_{22}}=\bmqty{J_{2d}&0\\0&J_{2d'}}.
    \end{align*} 
    By comparing the $(1,1)$-entry of both sides, we obtain the relation
    \begin{align*}
      L_{21}^TJ_{2d'}L_{21}=J_{2d}-X^TJ_{2d}X.
    \end{align*}
    Since $I_{2d'}+iJ_{2d'}\geq 0$, using the above identity, we get
    \begin{align*}
         0\leq L_{21}^T(I_{2d'}+iJ_{2d'})L_{21} 
          =L_{21}^TL_{21}+i(J_{2d}-X^TJ_{2d}X)
          =Y+i(J_{2d}-X^TJ_{2d}X).
    \end{align*}
    $(ii)\Rightarrow (i)$  Assume that $Y+i(J_{2d}-X^TJ_{2d}X)\geq 0$. 
    Then for any $\z\in\mbb{R}^{2d}$ we have 
     \begin{align*}
         0 \leq \mel{\z}{\big(Y+i(J-X^TJX\big)}{\z} =\mel{\z}{Y}{\z}+i\mel{\z}{(J-X^TJX)}{\z}. 
     \end{align*}
     This implies that $\mel{\z}{(J-X^TJX)}{\z}=0$ and $\mel{\z}{Y}{\z}\geq 0$ for all $\z\in\mbb{R}^{2d}$. Therefore $Y\geq 0.$
    Now we will construct $L\in M_{6d}(\BR)$, i.e., $d' = 2d$. 
    To do so, we proceed by splitting the remainder of the proof into two steps:\\
    
    \noindent\underline{\textsc{Step 1:}} To show that there exists $L_{21}\in M_{4d\times 2d}(\BR)$ such that $L_{21}^TJ_{4d}L_{21}=J_{2d}-X^TJ_{2d}X$ and $L_{21}^TL_{21}=Y$. We do this in two cases: \\    
    
    \ul{\textsc{Case 1:}} Assume that $Y$ is invertible. Then the operator $Y^{-\frac{1}{2}}(J_{2d}-X^TJ_{2d}X)Y^{-\frac{1}{2}}$ is skew-symmetric, and hence by \cite[Corollary 2.5.11]{Horn2012-xx}, 
    there exists a  positive diagonal matrix $D\in M_{d}(\mbb{R})$ and an orthogonal matrix $R\in M_{2d}(\mbb{R})$ such that
    \begin{align*}
     R^TY^{-\frac{1}{2}}(J_{2d}-X^TJ_{2d}X)Y^{-\frac{1}{2}}R=\bmqty{0&D\\-D&0}.
    \end{align*}
    Observe that $Y+i(J_{2d}-X^TJ_{2d}X)\geq 0 $ implies that $I+iR^TY^{-\frac{1}{2}}(J_{2d}-X^TJ_{2d}X)Y^{-\frac{1}{2}}R\geq 0$, which further implies $I+i\bmqty{0&D\\ -D&0}\geq 0$.
    Hence $D$ is a positive contraction. Therefore, by Corollary \ref{cor-intertwing-J and D}, there exists $Q\in M_{2d\times 4d}(\mbb{R})$ such that $QQ^T=I_{2d}$ and $QJ_{4d}Q^T=\bmqty{0&D\\-D&0}$. Let $L_{21}:=Q^TR^TY^{\frac{1}{2}}$. Then $L_{21}^TL_{21}=Y$ and
    \begin{align*}
        L_{21}^TJ_{4d}L_{21}
            &=Y^{\frac{1}{2}}R\bmqty{0&D\\-D&0}R^TY^{\frac{1}{2}}\\
            &=Y^{\frac{1}{2}}RR^TY^{-\frac{1}{2}}(J_{2d}-X^TJ_{2d}X)Y^{-\frac{1}{2}}RR^TY^{\frac{1}{2}}\\
            &=J_{2d}-X^TJ_{2d}X.
    \end{align*}
    
    \ul{\textsc{Case 2:}} Suppose $Y$ is not invertible. Note that $Y_n:=Y+\frac{1}{n}I$ is positive, invertible and satisfies $Y_n+i(J_{2d}-X^TJ_{2d}X)\geq 0$ for all $n\geq 1$. Therefore, by \textsc{Case 1},  there exist orthogonal matrices $R_n\in M_{2d}(\BR)$ and $Q_n\in M_{4d\times 2d}(\BR)$ such that $Q_nQ_n^{T}=I_{2d}$ and $L_{21,n}:=Q_n^TR_n^TY_n^{\frac{1}{2}}$ satisfies 
    \begin{align*}
        (L_{21,n})^TL_{21,n}=Y_n
        \qquad\mbox{and}\qquad
        L_{21,n}^TJ_{4d}L_{21,n}=J_{2d}-X^TJ_{2d}X,\qquad\forall~n\geq 1.
    \end{align*}
    Since both the sets  $\{Q\in M_{4d\times 2d}(\BR):QQ^{T}=I\}$  and $\{R\in M_{2d}(\BR): R \mbox{ is orthogonal}\}$ are compact, we get a subsequence $\{Q_{n_k}\}\subseteq\{Q_n\}$ and $\{R_{n_k}\}\subseteq\{R_n\}$ such that $Q_{n_k}$ converges to $Q$ and $R_{n_k}$ converges to $R$. Note that $QQ^T=I_{2d}$ and $R$ is orthogonal. Further, $L_{21,n_k}$ converges to $L_{21}:=Q^TR^TY^{\frac{1}{2}}$. From the construction $L_{21}^TL_{21}=Y$ and $L_{21}^TJ_{4d}L_{21}=J_{2d}-X^TJ_{2d}X$.\\
    
    \noindent\ul{\textsc{Step 2:}} To construct the required matrix $L$. 
    For all $1\leq i\leq d$, let $\xi_i:=(e_i,0)$ and $\eta_i=(0,e_i)\in\mbb{R}^{2d}$, where $0\in\mbb{R}^d$ and $\{e_i\}_{i=1}^d\subseteq\mbb{R}^d$ is the standard orthonormal basis. In $\mbb{R}^{2d}\oplus\mbb{R}^{4d}$, consider the sets 
    \begin{align*}
        E_1=\big\{\bu_i:=X\xi_i\oplus L_{21}\xi_i\big\}_{i=1}^d
        \quad\mbox{and}\qquad
        E_2=\big\{\bv_i:=X\eta_i\oplus L_{21}\eta_i\big\}_{i=1}^d.
    \end{align*}
    Note that 
    \begin{align*}
        \bmqty{X\\ L_{21}}=\bmqty{\bu_1&\bu_2&\cdots&\bu_d&\bv_1&\bv_2&\cdots&\bv_d}.
    \end{align*}
    Further, for all $1\leq i,k\leq d$ we have
    \begin{align*}
        \bu_i^T(J_{2d}\oplus J_{4d})\bu_k
        &=(X\xi_i\oplus L_{21}\xi_i)^T(J_{2d}\oplus J_{4d})(X\xi_k\oplus L_{21}\xi_k)\\
        &=\xi_i^T(X^TJ_{2d}X+L_{21}^TJ_{4d}L_{21})\xi_k\\
        &=\xi_i^TJ_{2d}\xi_k\\
        &=0.       
    \end{align*}
    Similarly, for all $1\leq i,k\leq d$, we see that
    \begin{align*}
       \bv_i^T(J_{2d}\oplus J_{4d})\bv_k=\eta_i^TJ_{2d}\eta_k=0
       \qquad\mbox{and}\qquad
       \bu_i^T(J_{2d}\oplus J_{4d})\bv_k=\xi_i^TJ_{2d}\eta_k=\delta_{i,k}.
    \end{align*}
    Observe that $J_{2d}\oplus J_{4d}=P^TJ_{6d}P$, where $P\in M_{6d}(\BR)$ is the permutation matrix given by
    \begin{align*}
        P=\bmqty{I_d&0&0&0\\0&0&I_d&0\\0&I_{2d}&0&0\\0&0&0&I_{2d}}.
    \end{align*}
    Then, for all $1\leq i,k\leq d$ we have
    \begin{align*}
        (P\bu_i)^TJ_{6d}(P\bu_k)&=\bu_i^T(J_{2d}\oplus J_{4d})\bu_k=0\\
        (P\bv_i)^TJ_{6d}(P\bv_k)&=\bv_i^T(J_{2d}\oplus J_{4d})\bv_k=0\\
        (P\bu_i)^TJ_{6d}(P\bv_k)&=\bu_i^T(J_{2d}\oplus J_{4d})\bv_k=\delta_{ik}.
    \end{align*}
    Hence, by \cite[Theorem 1.15]{Gosson2006-lr} and 
    \eqref{eq-sym-matrix-basis}, there exist vectors $\{\bu_j'\}_{j=d+1}^{3d},\{\bv_j'\}_{j=d+1}^{3d}\subseteq \BR^{6d}$ such that 
    \begin{align*}
        M:=\Big[P\bu_1~~\cdots~~P\bu_d~~\bu_{d+1}'~~\cdots~~\bu_{3d}'~~P\bv_1~~\cdots~~P\bv_d~~\bv_{d+1}'~~\cdots~~\bv_{3d}'\Big]\in M_{6d}(\BR)
    \end{align*}
    is a symplectic matrix. Consequently, $L:=P^TMP$ is a $(J_{2d}\oplus J_{4d})$-symplectic matrix. But 
    \begin{align*}
        P^TMP
        &=\Big[\bu_1~~\cdots~~\bu_d~~P^T\bu_{d+1}'~~\cdots~~P^T\bu_{3d}'~~\bv_1~~\cdots~~\bv_d~~P^T\bv_{d+1}'~~\cdots~~P^T\bv_{3d}'\Big]P\\
        &=\bmqty{\bu_1&\cdots&\bu_d&\bv_1&\cdots&\bv_d&\ast&\ast&\cdots&\ast}\\
        &=\bmqty{X&\ast\\ L_{21}&\ast},
    \end{align*}
    where '$\ast$' denotes unspecified block matrices.\\
    $(ii)\iff (iii)$ Assume that  $Z:=Y+i(J_{2d}-X^TJ_{2d}X)\geq 0$. Then we have that $Y\geq 0$ as noticed above. Since $J^T=-J$, the transpose matrix $Z^T=Y-i(J_{2d}-X^TJ_{2d}X)\geq 0$. Conversely, the reverse direction follows in a similar manner. This completes the proof. 
\end{proof}  

  Now we are ready to prove our main theorem.

\begin{theorem}\label{thm-Gaussian map-char}
    Let $\Psi:\TC{\Gamma(\mbb{C}^d)}\to\TC{\Gamma(\mbb{C}^d)}$ be a quantum channel. Then the following are equivalent: 
    \begin{enumerate}[label=(\roman*)]
        \item\label{item-Gaussian} The channel $\Psi$ maps Gaussian states to Gaussian states; 
        \item\label{item-Psi} There exist $X,Y\in M_{2d}(\BR)$ and $\w\in\BR^{2d}$ such that for every Gaussian state $\rho_{(\m,S)}$, the transformed state $\Psi(\rho_{(\m,S)})$ is also a Gaussian state with mean $X^T\m+\w$ and covariance $X^TSX+Y$, i.e., 
        \begin{align}\label{eq-QGC-state-image}
            \Psi(\rho_{(\m,S)})=\rho_{(X^T\m+\w,X^TSX+Y)};
        \end{align}
        \item\label{item-Psi-dual}  There exist $X,Y\in M_{2d}(\BR)$ and $\w\in\BR^{2d}$ such that  the dual map $\Psi^*: \B{\GCd}\to\B{\GCd}$ satisfies
        \begin{align}\label{eq-GQC-dualmap}
            \Psi^*(W(\z))=\exp{-i{\w}^T{\z}-\frac{1}{2}{\z}^T{Y\z}} W(X\z),\quad \forall~ \z\in\mbb{\BR}^{2d}; 
        \end{align}
        \item\label{item-Stine} There exists a Gaussian unitary operator ${U}$ on $\Gamma(\BC^{d+d'})= \Gamma(\BC^{d})\otimes \Gamma(\BC^{d'})$, for some $d'\geq 1$,  such that
        \begin{align}\label{eq-GQC-stinespring-1}
             \Psi(\varrho)=\tr_{2}\big({U}(\varrho\otimes \rho_{(0,I_{2d'})})U^{\dagger}\big), \qquad \forall~ \varrho\in\TC{\Gamma(\mbb{C}^d)};
        \end{align} 
        \item \label{item-ampli} The channel $id_{\TC{\Gamma(\mbb{C}^k)}}\otimes\Psi$ maps Gaussian states on $\Gamma(\mbb{C}^{k}\oplus \mbb{C}^{d})$ to Gaussian states for all $k\in\mbb{N}\cup\{0\}$. (Here $\Gamma(\BC^{0}):=\mbb{C}$.)
    \end{enumerate}
    In such a case, in both \ref{item-Psi} and \ref{item-Psi-dual}, the matrix $Y+i(J_{2d}-X^TJ_{2d}X)\geq 0$ and hence $Y\geq0$.
\end{theorem}

\begin{proof}
    $\ref{item-Gaussian}\implies\ref{item-Psi}$ Note that a Gaussian state is completely characterized by its mean and covariance matrix. Hence, by assumption, there exist  functions $\alpha_1:\BR^{2d}\times \msc{CM}(d)\to \BR^{2d}$ and $\alpha_2:\BR^{2d}\times \msc{CM}(d)\to \msc{CM}(d)$ such that for every Gaussian state $\rho_{(\m,S)}$, the transformed state $\Psi(\rho_{(\m,S)})$ is also a Gaussian state with mean $\alpha_1(\m,S)$ and covariance $\alpha_2(\m,S)$. The rest of the proof follows from Theorem \ref{thm-Polletti}. \\
    $\ref{item-Psi}\implies\ref{item-Stine}$ Assume that the channel $\Psi$ satisfies \eqref{eq-QGC-state-image}. Then, by Proposition \ref{prop-pos-CP-condn} in the Appendix, we have $Y+i(J_{2d}-X^TJ_{2d}X)\geq 0$. Hence, by Proposition \eqref{prop-symplectic-extension}, we get $L=\bmqty{X&L_{12}\\L_{21}&L_{22}}\in Sp(6d, J_{2d}\oplus J_{4d})$ such that $L_{21}^{T}L_{21}=Y$. Then $U:=W(\frac{1}{2}J_{2d}^T\w\oplus 0)\Gamma(L^{-1})\in\Gamma(\mathbb{C}^{3d})$ is a Gaussian unitary, where $0\in\BR^{4d}$. Define  $\Psi':\TC{\Gamma(\mbb{C}^d)}\to\TC{\Gamma(\mbb{C}^d)}$ by 
    \begin{align*}\Psi'(\varrho):=\tr_{2}\big(U(\varrho\otimes \rho_{(0,I_{4d}})U^\dagger\big), \qquad \forall~ \varrho\in\TC{\Gamma(\mbb{C}^d)}.
    \end{align*}
    Then $\Psi'$ is a quantum channel. It can be easily verified that
    \begin{align*}
        \tr(\Psi'(\rho_{(\m,S)})W(\z))=\tr(\Psi(\rho_{(\m,S)})W(\z))
    \end{align*}
    for all Gaussian states $\rho_{(\m,S)}$ and $\z\in\BR^{2d}$. Hence, we conclude that $\Psi'=\Psi$.\\
    $\ref{item-Stine}\implies\ref{item-ampli}$. By the assumption, there exists a Gaussian unitary operator ${U}$ on $\Gamma(\BC^{d+d'})= \Gamma(\BC^{d})\otimes \Gamma(\BC^{d'})$, for some $d'\geq 1$,  such that
    \begin{align}
             \Psi(\varrho)=\tr_{2}\big({U}(\varrho\otimes \rho_{(0,I_{2d'})})U^{\dagger}\big), \qquad \forall~ \varrho\in\TC{\Gamma(\mbb{C}^d)};
    \end{align}
    Fix $k\geq 1$. We show that there exists Gaussian unitary $\widetilde{U}$ on $\Gamma(\BC^{k+d+d'})=\Gamma(\BC^{k+d})\otimes\Gamma(\BC^{d'})$ such that
    \begin{align*}
            (id_{\TC{\GCk}}\otimes\Psi)(\widetilde{\rho})=\tr_{2}\Big(\widetilde{U}(\widetilde{\rho}\otimes\rho_{(0,I_{2d'})})\widetilde{U}^\dagger\Big),\qquad\forall~\widetilde{\rho}\in\TC{\BC^{k+d}},
    \end{align*}
    where $\tr_2$ is the partial trace over the second factor $\Gamma(\BC^{d'})$. Then, from Proposition \ref{prop-mean-covar-gaussian symmetry}, we conclude that $id_{\TC{\GCk}}\otimes\Psi$ maps Gaussian states to Gaussian states. To this end we consider the Gaussian unitary $\widetilde{U}:=I_{\Gamma(\BC^{k})}\otimes U$ defined on $\Gamma(\BC^{k})\otimes \Gamma(\BC^{d+d'})=\Gamma(\BC^{k+d+d'})$. Then for any $T_1\in\TC{\Gamma(\BC^{k})}$ and $T_2\in\TC{\Gamma(\BC^{d})}$,
    \begin{align*}
      \tr_{2}\big(\widetilde{U}(T_1\otimes T_2\otimes \rho_{(0,I_{2d'})})\widetilde{U}^{\dagger}\big)
      &=\tr_{2}\big((I_{\Gamma(\BC^{k})}\otimes U)(T_1\otimes T_2\otimes \rho_{(0,I_{2d'})})(I_{\Gamma(\BC^{k})}\otimes U)^{\dagger}\big)\\
      &=\tr_{2}\big(T_1\otimes U(T_2\otimes \rho_{(0,I_{2d'})})U^{\dagger} \big)\\
      &=T_1\otimes \tr_{2}\big(U(T_2\otimes \rho_{(0,I_{2d'})})U^{\dagger}\big)\\
      &=(id_{\TC{\Gamma(\mbb{C}^k)}}\otimes\Psi)(T_1\otimes T_2).
    \end{align*}
    Since the set $\{T_1\otimes T_2: T_1\in\TC{\Gamma(\BC^{k})},T_1\in\TC{\Gamma(\BC^{d})}\}$ is total, we conclude that 
    \begin{align*}
      (id_{\TC{\Gamma(\mbb{C}^k)}}\otimes\Psi)(\widetilde{\rho})=
      \tr_{2}\big(\widetilde{U}(\widetilde{\rho}\otimes \rho_{(0,I_{2d'})})\widetilde{U}^{\dagger}\big), \qquad\forall~ \widetilde{\rho}\in\TC{\Gamma(\BC^{k+d})}.
    \end{align*}
    $\ref{item-ampli}\Rightarrow\ref{item-Gaussian}$. Follows from the case $k=0$.\\
    $\ref{item-Psi}\iff\ref{item-Psi-dual}$ Assume that there exist $X,Y\in M_{2d}(\BR)$ such that $\Psi$ satisfies \eqref{eq-QGC-state-image}. Then, for any $\z\in \BR^{2d}$ and Gaussian states $\rho_{(\m,S)}$ consider
    \begin{align*}
     \tr\left(\rho_{(\m,S)}\Psi^*(W(\z))\right)&=\tr\left(\Psi(\rho_{(\m,S)})W(\z)\right)\\&= \widehat{\Psi(\rho_{(\m,S)})}(\z)\\
     &=\exp{-i\left({X^{T}\m+\w}\right)^T{\z}-\frac{1}{2}{\z}^T{(X^{T}SX+Y)\z}}\\
     &=\exp{-i\w^T{\z}-\frac{1}{2}\z^TY\z}\exp{-i\m^TX\z-\frac{1}{2}\z^TX^TSX\z}\\
     &=\exp{-i{\w}^T{\z}-\frac{1}{2}{\z}^T{Y\z}}\tr\big(\rho_{(\m,S)}W(X\z)\big).
    \end{align*}
    Since the span of Gaussian states is dense in $\TC{\Gamma(\mbb{C}^d)}$ by \cite[Theorem II.4]{Fannes1975-dh}, (cf. \cite{Fannes1976-yg}), we get 
    \begin{align*}
     \tr\big(\varrho\Psi^*(W(\z))\big)=\exp{-i{\w}^T{\z}-\frac{1}{2}{\z}^T{Y\z}}\tr\big(\varrho W(X\z)\big),\qquad\forall~\varrho\in\TC{\Gamma(\mbb{C}^d)}.
    \end{align*}
    Hence
    \begin{align*}
     \Psi^*(W(\z))=\exp{-i{\w}^T{\z}-\frac{1}{2}{\z}^T{Y\z}} W(X\z), \qquad \forall~ \z\in\BR^{2d}.
    \end{align*}
    Conversely, if $\Psi^*$ satisfies \eqref{eq-GQC-dualmap}, then for all Gaussian states $\rho_{(\m,S)}$ and $\z\in\mbb{R}^{2d}$ we have
    \begin{align*}
        \tr(\Psi(\rho_{(\m,S)})W(z))
           &=\tr(\rho_{(\m,S)}\Psi^*(W(z)))\\
           &=\exp\{-i(X^T\m+\w)^Tz-\frac{1}{2}\z^T(X^TSX+Y)\z\}\\
           &=\tr(\rho_{(X^T+\m,X^TSX+Y)}W(z)),
    \end{align*}
    so that $\Psi$ satisfies \eqref{eq-QGC-state-image}.
    This completes the proof.
\end{proof}

    A quantum channel $\Psi:\TC{\GCd}\to\TC{\GCd}$ that satisfies any of the equivalence conditions in the above theorem is called a \textit{($d$-mode) quantum Gaussian channel}. Quantum channels $\Psi$ of the form \eqref{eq-GQC-stinespring-1} or equivalently \eqref{eq-GQC-stinespring} are commonly referred to as Gaussian channels among physicists. Additionally, quantum channels that satisfy the condition $\ref{item-Psi-dual}$ of the above theorem are also referred to as  Gaussian channels in the literature. 

\begin{rmk}
    Let $\Psi:\TC{\GCd}\to\TC{\GCd}$ be a quantum channel that satisfies any of the statements \ref{item-Psi}-\ref{item-Psi-dual} of the above theorem. Then the above proof indicates that we can choose the Gaussian unitary operator in \ref{item-Stine} to be defined on $\Gamma(\BC^{3d})=\Gamma(\BC^{d})\otimes\Gamma(\BC^{2d})$. However, we have not discussed the minimality of $d'$, i.e., the minimal number of environment modes required for  the dilation, in Theorem \ref{thm-Gaussian map-char} \ref{item-Stine}. Although a complete characterization of the minimal number of environment modes required for a dilation is not available in the literature, an upper bound for $d'$ is given in \cite[Theorem 2]{Caruso2008-ac} as
    \begin{align*}
        d'=\text{rank}[Y]-\text{rank}[Y-\Sigma Y^{\ominus}\Sigma^T],
    \end{align*}
    where $Y^{\ominus}$ is the Moore-Penrose inverse of $Y$, $\Sigma = J_{2d}-X^TJ_{2d}X$ and $X, Y$ are as in \eqref{eq-QGC-state-image}. 
\end{rmk}

\begin{rmk} We give a proof of  $\ref{item-Stine}\Rightarrow \ref{item-Psi}$ here because it shows how to get the channel parameters $X$ and $Y$ from the Stinespring representation of a Gaussian channel. 
    Suppose $\Psi$ is a quantum Gaussian channel given by \eqref{eq-GQC-stinespring-1}. Then, by Proposition \ref{prop:gaussian-unitary}, there exist $\lambda\in\BC$ with $\abs{\lambda}=1$,  $\bu=\bu_1\oplus\bu_2\in\mbb{R}^{2d}\oplus\mbb{\BR}^{2d'}$ and $L\in Sp(2(d+d'), J_{2d}\oplus J_{2d'})$ such that  $U= \lambda W(\bu)\Gamma(L)$. Without loss of generality, assume that $\lambda=1$. Write $L^{-1}=\bmqty{L_{11}&L_{12}\\L_{21}&L_{22}}$ with $L_{11}\in M_{2d}(\BR)$ and $L_{22}\in M_{2d'}(\BR)$. Then, by Proposition \ref{prop-mean-covar-gaussian symmetry}, we have 
    \begin{align*}
       \Psi(\rho_{(\m,S)})=\rho_{(X^{T}\m+\w, X^{T}SX+Y)},
    \end{align*}
    is a Gaussian state for any Gaussian state $\rho_{(\m,S)}$, where $\w ={2}J_{2d}\bu_1, X=L_{11}$ and $Y=L_{21}^{T}L_{21}$. 
\end{rmk}
    
\begin{cor}\label{cor-KRP-que}
    Let $X,Y\in M_{2d}(\mbb{R})$. Then the following are equivalent:
    \begin{enumerate}[label=(\roman*)]
        \item There exists a quantum Gaussian channel $\Psi:\TC{\GCd}\to\TC{\GCd}$ that transforms a Gaussian state $\rho_{(\m,S)}$ to a Gaussian state with covariance $X^TSX+Y$; 
        \item $Y+i(J_{2d}-X^TJ_{2d}X)\geq 0$.
    \end{enumerate}
\end{cor}   

\begin{proof}
    This follows directly from Theorem \ref{thm-Gaussian map-char} and the proof of the implication \ref{item-Psi} $\implies$ \ref{item-Stine}.
\end{proof}

\section{Characterization of Linear Optical Channels}\label{sec:optics}

 A general \textit{$d$-mode linear optical channel}, also known as a \textit{$d$-mode universal multiport interferometer}, is mathematically described as a Gaussian channel $\Psi:\GCd\to\GCd$ of the form 
 \begin{equation}\label{eq:interferometer}
    \Psi(\varrho) = \tr_2\left(\Gamma(L)\left(\varrho\otimes\rho_{(0,I_{d'})}\right)\Gamma(L)^{\dagger}\right), \qquad \forall~ \varrho\in \TC{\Gamma(\mbb{C}^d)},
 \end{equation}
 where $L\in M_{2(d+d')}(\BR)$ is an orthosymplectic matrix and $d'\geq 1$. Here $d'$ is called the \textit{number of environment modes}. It is known that any such channel can be physically implemented using $d(d-1)/2$ Mach-Zehnder interferometers, that is, an optical device composed of only two $50:50$ beam splitters and two phase shifters \cite{Reck1994-un, Clements2016-hj}. As discussed in the previous section,  such a Gaussian channel is completely described by two matrices $X$ and $Y$. Since linear optical channels are special cases of Gaussian channels, it is important to understand when a Gaussian channel, represented by the matrices $X$ and $Y$, can be physically implemented using a multiport interferometer. We answer this question in the present section. The mathematical result at the heart of this physical problem is an orthosymplectic version of Proposition \ref{prop-symplectic-extension}, which we will prove next. First, we recall a simple linear algebraic fact.

\begin{lemma}\label{lem-Ortho-Sym-mat}
    A matrix $L\in M_{2d}(\BR)$ is an orthosymplectic matrix if and only if $L= \bmqty{A&B\\-B&A}$ for some $A, B \in M_d(\BR)$ with $A^TA+B^TB = I$ and $A^TB$ is symmetric.
\end{lemma}


\begin{lemma}\label{lem:orthosymplectic}
    Let $X,Y\in M_{d}(\BR)$ with $\norm{X}\leq 1$ and $Y\geq 0$. Then the following are equivalent:
    \begin{enumerate}[label=(\roman*)]
        \item There exists an orthosymplectic matrix $L=\bmqty{X&B\\-B&X}\in M_{2d}(\BR)$ with $ B^TB=Y$;
        \item $X^TX+Y = I$ and $X^TQ\sqrt{Y}$ is symmetric for some orthogonal matrix $Q\in M_{d}(\BR)$.
    \end{enumerate} 
\end{lemma}

\begin{proof}
    $(i)\Rightarrow (ii)$ Assume that there exists an orthosymplectic matrix $L=\bmqty{X&B\\-B&X}\in M_{2d}(\BR)$ with $ B^TB=Y$. Then, by the above lemma,  $X^TX+Y=I$ and $X^TB$ is symmetric. But, by polar decomposition, $B=Q\sqrt{B^TB}=Q\sqrt{Y}$ for some orthogonal matrix $Q\in M_{d}(\BR)$. \\
    $(ii)\implies (i)$ Clearly, $L=\bmqty{X&Q\sqrt{Y}\\-Q\sqrt{Y}&X} \in M_{2d}(\BR)$ is an orthosymplectic matrix with the required properties. 
\end{proof}

\begin{theorem}
    Let $\Psi:\GCd\to\GCd$ be a quantum Gaussian channel. Then the following are equivalent: 
    \begin{enumerate}
        \item The channel $\Psi$ can be implemented using a multiport interferometer with $d$ number of environment modes;
        \item There exist matrices $X,Y\in M_{2d}(\BR)$ such that $X^TX+Y = I$, $X^TQ\sqrt{Y}$ is symmetric for some orthogonal matrix $Q\in M_{2d}(\BR)$ and 
        \begin{align}\label{eq-QGC-state-image-1}
            \Psi(\rho_{(\m,S)})=\rho_{(X^T\m,X^TSX+Y)}
        \end{align}
        for all Gaussian states $\rho_{(\m,S)}$.
    \end{enumerate} 
\end{theorem}
\begin{proof}
    $(i)\Rightarrow (ii)$ Assume that $\Psi$ can be implemented by using a multiport interferometer with $d$ number of environment modes, i.e., there exists an orthosymplectic matrix $L\in M_{2(d+d)}(\BR)$ such that $\Psi$ is given by \eqref{eq:interferometer}. Since $L^{-1}\in M_{4d}(\BR)$ is orthosympletic,  by  Lemma \ref{lem-Ortho-Sym-mat}, we have $L^{-1}=\bmqty{X&B\\ -B&X}$ for some $X,B\in M_{2d}(\BR)$ satisfying $X^TX+B^TB=I$ and $X^TB$ is symmetric. Letting $Y:=B^TB$,  by Lemma \ref{lem:orthosymplectic}, we have $X^TQ\sqrt{Y}$ is symmetric for some orthogonal matrix $Q\in M_{2d}(\BR)$. Furthermore, by Proposition \ref{prop-mean-covar-gaussian symmetry}, it follows that $\Psi$ satisfies \eqref{eq-QGC-state-image-1}.\\
    $(ii)\Rightarrow (i)$ Conversely, assume that there exists an orthogonal matrix $Q\in M_{2d}(\BR)$ such that $X^TX+Y = I$ and $X^TQ\sqrt{Y}$ is symmetric and $\Psi$ satisfies \eqref{eq-QGC-state-image-1}. Then, $L^{-1}:=\bmqty{X&Q\sqrt{Y}\\-Q\sqrt{Y}&X} \in M_{4d}(\BR)$ is an orthosymplectic matrix. By Proposition \ref{prop-mean-covar-gaussian symmetry}, the map $\Psi'$ defined by
    \begin{align*}
        \Psi'(\varrho):=\tr_2(\Gamma(L)(\varrho\otimes\rho_{(0,I_{d})})\Gamma(L)^\dagger) \qquad \forall~ \varrho\in \TC{\GCd}
    \end{align*}
    defines a quantum Gaussian channel such that 
    \begin{align*}
        \Psi'(\rho_{(\m,S)})=\rho_{(X^Tm,X^TSX+Y)}
    \end{align*}
    for all Gaussian states $\rho_{(\m,S)}$. Since Gaussian states are total we conclude that $\Psi=\Psi'$. 
\end{proof}

\section{Conclusion}

This article provides a mathematically rigorous treatment of the equivalence of various definitions of Gaussian channels. Our analysis helps us to derive necessary and sufficient conditions for the physical implementation of  certain Gaussain channels using linear optical devices. 
Furthermore, it also answers  the questions asked by Parthasarathy \footnote{Note that Parthasarathy  uses $L^2(\BR^n)$ while we use $\Gamma(\BC^n)$.  The theory of Bosonic systems which we consider in this article remains same on both these space via the isomorphism explained in Section II item 4 of \cite{TiKR21}} \cite[page 438]{KRP15}:
Consider the following convex sets:
    \begin{align*}
        \mathscr{F}_d(X) & =\left\{Y : Y \geq 0, X^T S X+ Y \in\msc{CM}(d) \quad \forall~ S \in \msc{CM}(d)\right\}, \\
        \mathscr{F}_d^0(X) & =\left\{Y : Y+i\left(J_{2d}-X^T J_{2 d} X\right) \geq 0\right\},
    \end{align*}
    where $X\in M_{2d}(\BR)$.
    Note that elements of $\msc{CM}(d)$ and $\mathscr{F}_d^0(X)$ are necessarily positive. The following questions were asked by Parthasarathy:
    \begin{enumerate}
        \item 
         What are the necessary and sufficient conditions for the existence of a "symplectic dilation" $L=\bmqty{X&L_{12}\\ L_{21}&L_{22}}$ such that $L_{21}^{T}L_{21}\in \mathscr{F}_d^0(X)$?
        \item 
         Is it true that for every $Y\in \mathscr{F}_d(X)$ there exists a quantum channel that maps a Gaussian state $\rho_{(\m,S)}$ to a Gaussian state with covariance matrix $X^TSX+Y$? 
         \item Are there any Gaussian channels not belonging to the semigroup generated by all reversible, Bosonic, symplectic and quasifree Gaussian channels? (See \cite{KRP15} for details.)
\end{enumerate}
Proposition \ref{prop-symplectic-extension}  affirmatively answers the question (i). We will now show that the answer to question (ii) is \textbf{not true} in general. Note that, by Corollary \ref{cor-KRP-que}, there exists a Gaussian channel that transforms a Gaussian state $\rho_{(\m,S)}$ to a Gaussian state with covariance $X^TSX+Y$ if and only if $Y\in \mathscr{F}_d^0(X)$. 
The following example, inspired by one provided by Poletti through private communication in a different context, provides a pair $(X,Y)$ such that $Y\in \mathscr{F}_d(X)$ but $Y\notin \mathscr{F}_d^0(X)$.
    \begin{eg}
        Take $X=\bmqty{0&I_d\\I_{d}&0}$, $Y=\bmqty{I_{d}&0\\0&I_{d}}$. Then $Y\in\msc{CM}(d)$, i.e., $Y+iJ_{2d}\geq 0$. If $S\in\msc{CM}(d)$ then $S\geq0$ and hence
    \begin{align*}
        X^TSX+Y+iJ_{2d}\geq 0,
    \end{align*}
 i.e.,  $X^TSX+Y\in\msc{CM}(d)$ implying that $Y\in \mathscr{F}_d(X)$. On the other hand, $Y\notin \mathscr{F}_d^0(X)$ because 
    \begin{align*}
        Y+i\left(J_{2d}-X^T J_{2 d} X\right)&=\bmqty{I_{d}&0\\0&I_{d}}+i\left(\bmqty{0&I_d\\-I_{d}&0}-\bmqty{0&I_d\\I_{d}&0}\bmqty{0&I_d\\-I_{d}&0}\bmqty{0&I_d\\I_{d}&0}\right)\\
        &=\bmqty{I_{d}&2iI_{d}\\-2iI_{d}&I_{d}}\ngeq 0.
    \end{align*}
    \end{eg}
Now, coming to question (iii), the answer is \textbf{no}, because  from Theorem \ref{thm-Gaussian map-char} \ref{item-Stine}, and using the definitions of Parthasarathy we see that \textbf{every Gaussian channel is a symplectic Gaussian channel composed with a displacement, which is a reversible channel}, and hence it is in the semigroup generated by  all reversible, Bosonic, symplectic and quasifree Gaussian channels.

\section*{Acknowledgment}

RD is supported by the  Postdoctoral Fellowship of Indian Institute of Technology Bombay. TCJ expresses gratitude to Saikat Guha for his guidance in learning quantum optics and to the Indian Institute of Technology, Madras, for facilitating his visit to conduct this research.  SK acknowledges partial support from the IoE Project of MHRD (India)  under reference number  SB22231267MAETWO008573,  and is also partially funded by the NBHM grant (No. 02011/10/2023 NBHM (R.P) R\&D II/4225) for this research.

  \appendix  

  \section{Supplementary Analysis of Existing Literature}
  A  trace-norm continuous linear map $\Psi:\TC{\GCd}\to\TC{\GCd}$ is said to be\textit{ positive}, if $\Psi^*$ is positive.  Let $\Psi:\TC{\GCd}\to\TC{\GCd}$ be a positive trace preserving map. In \cite{Pol22}, Poletti called such a map to be a \textit{Gaussian map}  if there exist continuous maps $\alpha_1:\BR^{2d}\times \msc{CM}(d)\to \BR^{2d}$ and $\alpha_2:\BR^{2d}\times \msc{CM}(d)\to \msc{CM}(d)$ such that for every Gaussian state $\rho_{(\m,S)}$ on $\GCd$, the transformed state $\Psi(\rho_{(\m,S)})$ is also a Gaussian state with mean $\alpha_1(\m,S)$ and covariance $\alpha_2(\m,S)$, that is,
    \begin{align}\label{eq-Poletti-hypo}
        \Psi(\rho_{(\m,S)})=\rho_{(\alpha_1(\m,S),\alpha_2(\m,S))},\qquad\forall~(\m,S)\in\mbb{R}^{2d}\times\mathscr{CM}(d).
    \end{align}
    Furthermore, Poletti proved \cite[Theorem 4.1]{Pol22} that if $\Psi$ is a Gaussian map, then there exist $\w\in\mbb{R}^{2d}$ and $X,Y\in M_{2d}(\mbb{R})$ such that 
    \begin{align*}
        \alpha_1(\m,S)=X^T\m+\w,
        \quad
        \alpha_2(\m,S)=X^TSX +Y,
        \quad\mbox{and}\quad 
        Y+i(J_{2d}-X^TJ_{2d}X)\geq 0.
    \end{align*}
    (We have seen that the last condition implies $Y\geq 0$). However, in Proposition \ref{prop-conti-mean and cova} below, we show that any functions $\alpha_j$'s satisfying \eqref{eq-Poletti-hypo} are necessarily continuous, making the explicit continuity assumption superfluous. Furthermore, in Proposition \ref{prop-pos-CP-condn} below, we observe that if there exist $X,Y$ satisfying the inequality $Y+i(J_{2d}-X^TJ_{2d}X)\geq 0$, then $\Psi$ must be a quantum channel. In contrast, Poletti assumed that $\Psi$ is only a positive map. In fact, as demonstrated in Example \ref{ex:non-cp-gaussian}, the transpose map on $\TC{\GCd}$ provides a counterexample: it is a positive map that sends Gaussian states to Gaussian states via continuous maps $\alpha_j$ satisfying \eqref{eq-Poletti-hypo}, yet it is not a quantum channel. This is because its dual, the transpose map on $\B{\GCd}$, is not completely positive (cf.\ \cite{De_Palma2015-lz}). So we refine Poletti's theorem \cite[Theorem 4.1]{Pol22} as follows, using essentially the same proof: 

\begin{theorem}\cite{Pol22}\label{thm-Polletti}
    Let $\Psi:\TC{\GCd}\to\TC{\GCd}$ be a quantum channel. Let $\alpha_1:\mbb{R}^{2d}\times \mathscr{CM}(d)\to\mbb{R}^{2d}$ and $\alpha_2:\mbb{R}^{2d}\times \mathscr{CM}(d)\to\mathscr{CM}(d)$ be two functions such that \eqref{eq-Poletti-hypo} holds.
    Then there exist $\w\in\mbb{R}^{2d}$ and $X,Y\in M_{2d}(\mbb{R})$ such that 
    \begin{align*}
        \alpha_1(\m,S)=X^T\m+\w\quad\mbox{and}\quad\alpha_2(\m,S)=X^TSX+Y.
    \end{align*}
    Furthermore, $Y+i(J_{2d}-X^TJ_{2d}X)\geq 0$.
\end{theorem}    

    In the following, we address certain gaps in the original formulation of Poletti's result. First, we recall the following well-known result called Wigner isomorphism. For a mathematically oriented exposition of Wigner isomorphism one may refer to \cite[Theorem III.5]{TiKR21}.

 \begin{theorem}[\textbf{Wigner isomorphism}]\label{thm:wigner-iso}
       Let $L^2(\BC^d)$ denote the Hilbert space of square integrable functions on $\BC^d$ with respect to the $2d$-dimensional Lebesgue measure on $\BC^d$. Define $\mathbb{F}_d: \TC{\GCd}\to L^2(\BC^d)$, to be the map satisfying 
    \begin{equation}\label{eq:defn-Fn}
        \mathbb{F}_d(\rho)(\z) = \pi^{-d/2}\widehat{\rho}(\z), \qquad\forall~ \z \in \BC^d, \quad\rho\in\TC{\GCd}.
    \end{equation}
    Then $\mathbb{F}_d$ extends uniquely to a Hilbert space isomorphism, again denoted by $\mathbb{F}_d$, from $\mathscr{T}_2({\GH}$) onto $L^2(\BC^d)$, where $\mathscr{T}_2({\GH}$ denotes the Hilbert space  of Hilbert-Schmidt operators on $\GH$.
    \end{theorem}
    
\begin{proposition}\label{prop-weak-trace norm-conv}
    Let $\rho,\rho_n$ be states on $\GCd$. Then $\rho_n\to\rho$ in trace norm if and only if $\tr(\rho_nW(z))\to\tr(\rho W(z))$ for all $z\in\mbb{C}^d$.
\end{proposition}

\begin{proof}
    Assume that $\tr(\rho_nW(z))\to\tr(\rho W(z))$ for all $z\in\mbb{C}^d$. To show that $\rho_n\to \rho$ in trace-norm. By \cite[Theorem in Appendix]{Arazy1981-ce}, it is enough to prove that $\rho_n\to \rho$ in weak operator topology.  For any $\z,\w\in\mbb{C}^d$, by Theorem \ref{thm:wigner-iso}, we have 
    \begin{align}
        \mel{e(\w)}{\rho}{e(\z)} 
            & = \tr \ketbra{e(\z)}{e(\w)} \rho\nonumber\\
            &= \tr (\ketbra{e(\w)}{e(\z)}^\dagger\rho)\nonumber\\ 
            & = \int\limits_{\BC^n}^{}\overline{ \mathbb{F}_d( \ketbra{e(\w)}{e(\z)})(\bu) }\mathbb{F}_d(\rho)\dd \bu\nonumber\\
            &= \int\limits_{\BC^n}^{}\overline{ \mathbb{F}_d( \ketbra{e(\w)}{e(\z)})(\bu) }\pi^{-d/2}\widehat{\rho}(\bu)\dd \bu\nonumber\\
            & =\int\limits_{\BC^d}^{}\overline{ \mathbb{F}_d( \ketbra{e(\w)}{e(\z)})(\bu) }\pi^{-d/2}\left(\lim_{k\rightarrow\infty} \widehat{\rho}_k(\bu)\right)\dd \bu.\label{eq:norm-iff-qcf-pointwise}
    \end{align}
    Now we will use the dominated convergence theorem to interchange the limit and the integration above. To this end, note that for every $\z,\w$, 
    \begin{align*}
        \pi^{-d/2}\mathbb{F}_d(\ketbra{e(\w)}{e(\z)})(\bu)&= \tr \ketbra{e(\w)}{e(\z)}W(\bu)\\
            &= \mel{e(\z)}{W(\bu)}{e(\w)}\\
            &= e^{-\frac{1}{2}\abs{\bu}^2-\braket{\bu}{\w}}\braket{e(\z)}{e(\bu+\w)}\\
            &= e^{\bar{\bv}\bu}e^{-\frac{1}{2}|\z|^2+\bar{\bv}\z-\bar{\z}\bu}. 
    \end{align*}
    Thus, the function  $\bu\mapsto \mathbb{F}_d( \ketbra{e(\w)}{e(\z)})(\bu)$ is integrable. By the dominated convergence theorem, we have
    \begin{align*}
     \int\limits_{\BC^d}^{}\overline{ \mathbb{F}_d( \ketbra{e(\w)}{e(\z)})(\bu) }\pi^{-d/2}\left(\lim_{k\rightarrow\infty} \widehat{\rho}_k(\bu)\right)\dd \bu 
        &=\lim_{k\rightarrow\infty}\int\limits_{\BC^n}^{}\overline{ \mathbb{F}_n( \ketbra{e(\w)}{e(\z)})(\bu) }\pi^{-d/2} \widehat{\rho}_k(\bu)\dd \bu\\
        &= \lim_{k\rightarrow\infty} \mel{e(\w)}{\rho_k}{e(\z)}. 
    \end{align*}
    Finally, combining \eqref{eq:norm-iff-qcf-pointwise} and the above we have proved that 
    \begin{align}
        \mel{e(\w)}{\rho}{e(\z)} = \lim_{k\rightarrow\infty} \mel{e(\w)}{\rho_k}{e(\z)},\quad \forall \z,\w\in \BC^d.
    \end{align}
    Now, the totality of exponential vectors shows that $\rho_k\to \rho$ in the weak operator topology.
\end{proof}

\begin{proposition}\label{prop-conti-mean and cova}
   Let $\Psi:\TC{\Gamma(\mbb{C}^d)}\to\TC{\Gamma(\mbb{C}^d)}$ be a positive map and let  $\alpha_1:\BR^{2d}\times \msc{CM}(d)\to \BR^{2d}$ and $\alpha_2:\BR^{2d}\times \msc{CM}(d)\to \msc{CM}(d)$ be such that for every Gaussian state $\rho_{(\m,S)}$, the transformed state $\Psi(\rho_{(\m,S)})$ is also a Gaussian state with mean $\alpha_1(\m,S)$ and covariance $\alpha_2(\m,S)$. Then the maps $\alpha_1,\alpha_2$ are continuous.
\end{proposition}

\begin{proof}
    Let $(\m_n,S_n),(\m,S)\in\BR^{2d}\times \msc{CM}(d)$ be such that $(\m_n,S_n)\to (\m,S)$. Define the Gaussian states $\rho:=\rho(\m,S)$ and $\rho_n:=\rho(\m_n,S_n)$ for all $n\geq 1$ on ${\Gamma(\mbb{C}^d)}$. Then, for all $\z\in\mbb{C}^d$,
    \begin{align*}
        \lim_n\widehat{\rho_{n}}(\z)&=\lim_n \exp\{-i\Re{\braket{\m_n}{\z}}-\frac{1}{2}\Re{\mel{\z}{S_n}{\z}}\}\\
                                   &=\exp\{-i\Re{\braket{\m}{\z}}-\frac{1}{2}\Re{\mel{\z}{S}{\z}}\}\\
                                   &=\widehat{\rho}(\z).
    \end{align*}
    Then, by Proposition \ref{prop-weak-trace norm-conv}, we obtain $\rho_{n}\to\rho$ in the trace norm. As $\Psi$ is continuous, we get $\Psi(\rho_{n})\to\Psi(\rho)$ in trace norm. Let $\sigma:=\Psi(\rho)$ and $\sigma_n:=\Psi(\rho_n)$ for all $n\geq 1$. Note that 
    $\sigma$ is a Gaussian state with mean $\alpha_1(\m,S)$ and covariance $\alpha_2(\m,S)$, and similarly $\sigma_n$ is also a Gaussian state with mean $\alpha_1(\m_n,S_n)$ and covariance $\alpha_2(\m_n,S_n)$ for all $n\geq 1$. Again, by Proposition \ref{prop-weak-trace norm-conv}, for all $z\in\mbb{C}^d$ we have $\widehat{\sigma_n}(z)\to\widehat{\sigma}(z)$, i.e., 
    \begin{align*}
        \lim_n &\exp\{-i\Re{\braket{\alpha_1(\m_n,S_n)}{\z}}-\frac{1}{2}\Re{\mel{\z}{\alpha_2(\m_n,S_n)}{\z}}\}\\
        =&\exp\{-i\Re{\braket{\alpha_1(\m,S)}{\z}}-\frac{1}{2}\Re{\mel{\z}{\alpha_2(\m,S)}{\z}}\}.
    \end{align*} 
    Fix $\z\in\mbb{C}^d$. Then by applying Logarithm to both sides, 
    and then comparing the  real and imaginary part, we get
    \begin{align*}
        &\lim_n\Re{\braket{\alpha_1(\m_n,S_n)}{\z}}=\Re{\braket{\alpha_1(\m,S)}{\z}}
        \quad\text{and}\quad \\
        &\lim_n\Re{\mel{\z}{\alpha_2(\m_n,S_n)}{\z}}=\Re{\mel{\z}{\alpha_2(\m,S)}{\z}}.
    \end{align*}
    Since $\z\in\mbb{C}^d$ is arbitrary, we conclude that $\alpha_1(\m_n,S_n)\to \alpha_1(\m,S)$ and $\alpha_2(\m_n,S_n)\to \alpha_2(\m,S)$. This completes the proof. 
    \end{proof}
   
\begin{proposition}\label{prop-pos-CP-condn}
    Let $\Psi:\TC{\GCd}\to\TC{\GCd}$ be a positive trace-preserving map such that 
    \begin{align*}
        \Psi(\rho_{(\m,S)})=\rho_{(X^T\m+\w,X^TSX+Y)},\qquad\forall~\rho_{(\m,S)} 
    \end{align*}
    and for some $X,Y\in M_{2d}(\BR)$ and $\w\in\BR^{2d}$. Then $\Psi$ is a quantum channel if and only if $Y+i(J_{2d}-X^TJ_{2d}X)\geq 0$. 
 \end{proposition}
 
  For a proof of the above proposition, the reader is referred to \cite[Section 3]{Parthasarathy_2022}.
\begin{eg}\label{ex:non-cp-gaussian}
    Let $\mathcal{B}$ be the \textit{Fock number basis} (also known as \textit{particle basis} see \cite[Section II]{TiKR21} cf. \cite[Proposition 19.3]{Krp92}) for the space $\GCd$. Let $\Psi:\TC{\GCd}\to \TC{\GCd}$ be the transpose map with respect to the basis $\mathcal{B}$, i.e., $\Psi(\rho):=\rho^T$, which is defined so that the matrix representation of $\rho^T$ in the basis $\mathcal{B}$ is the transpose of the matrix representation of $\rho$ in the basis $\mathcal{B}$.  By \cite[Proposition V.5 and Theorem V.7]{TiKR21}, a state $\rho$ on $\GCd$ is Gaussian if and only if there exists a triple $(\bm{\alpha}, A, \Lambda)$ consisting of a scalar $\bm{\alpha}\in\mbb{C}$, and matrices $A, \Lambda\in M_d(\mbb{C})$ satisfying $A=A^T$ and $\Lambda\geq 0$, such that
    \begin{align}\label{eq:tcj-krp}
        G_\rho(\bu, \bv):= \mel{e(\bar{\bu})}{\rho}{e(\bv)}= c\exp \left\{\mathbf{u}^T \boldsymbol{\alpha}+\bar{\boldsymbol{\alpha}}^T \mathbf{v}+\mathbf{u}^T A \mathbf{u}+\mathbf{u}^T \Lambda \mathbf{v}+\mathbf{v}^T \bar{A} \mathbf{v}\right\},
    \end{align} for all $ \bu,\bv\in \BC^d$, where $c=\mel{e(0)}{\rho}{e(0)}$.
    Note that $\mel{e(\bar{\bu})}{\rho^T}{e(\bv)} = \mel{e(\bar{\bv})}{\rho}{e(\bu)}$. So 
    \begin{align*}
      G_{\rho^T}(\bu, \bv)  =\mel{e(\bar{\bv})}{\rho}{e(\bu)}&=c \exp \left\{\mathbf{v}^T\boldsymbol{\alpha}+\bar{\boldsymbol{\alpha}}^T\mathbf{u}+\mathbf{v}^T A  \mathbf{v}+\mathbf{v}^T \Lambda \mathbf{u}+\mathbf{u}^T  \bar{A}\mathbf{u}\right\}\\
      &=c \exp \left\{\mathbf{u}^T\bar{\boldsymbol{\alpha}}+\boldsymbol{\alpha}^T\mathbf{v}+\mathbf{u}^T  \bar{A}\mathbf{u}+\mathbf{u}^T \Lambda^T \mathbf{v}+\mathbf{v}^T A  \mathbf{v}\right\}, 
    \end{align*}
    which preserves the structure of the right side of \eqref{eq:tcj-krp} with the new triple $(\bar{\bm{\alpha}}, \bar{A}, \Lambda^T)$. Hence, the transpose map $\Psi$  sends Gaussian states to Gaussian states. Note now that the dual map $\Psi^*$ is the transpose map of $\B{\GCd}$ with respect to the same basis $\mathcal{B}$. It is well-known that the transpose map is a positive but not CP map. By Proposition \ref{prop-conti-mean and cova}, we see that the transpose map satisfies the hypothesis of \cite[Theorem 4.1]{Pol22} and \cite[Theorem 2]{Poletti2022-wk}  but transpose map is not CP and hence not a quantum channel.
\end{eg}

\bibliographystyle{alpha}
\bibliography{ref.bib}

\newcommand{\etalchar}[1]{$^{#1}$}
\begin{thebibliography}{WPGP{\etalchar{+}}12}

\bibitem[Ara81]{Arazy1981-ce}
Jonathan Arazy.
\newblock {More on convergence in unitary matrix spaces}.
\newblock {\em Proc. Am. Math. Soc.}, 83(1):44, September 1981.

\bibitem[Att13]{AttQ13}
St{\'e}phane Attal.
\newblock Quantu channels.
\newblock {\em Lecture notes}, 2013.

\bibitem[CEGH08]{Caruso2008-ac}
F~Caruso, J~Eisert, V~Giovannetti, and A~S Holevo.
\newblock {Multi-mode bosonic Gaussian channels}.
\newblock {\em New J. Phys.}, 10(8):083030, August 2008.

\bibitem[CHM{\etalchar{+}}16]{Clements2016-hj}
William~R Clements, Peter~C Humphreys, Benjamin~J Metcalf, W~Steven Kolthammer, and Ian~A Walsmley.
\newblock {Optimal design for universal multiport interferometers}.
\newblock {\em Optica}, 3(12):1460, December 2016.

\bibitem[DPMGH15]{De_Palma2015-lz}
Giacomo De~Palma, Andrea Mari, Vittorio Giovannetti, and Alexander~S Holevo.
\newblock {Normal form decomposition for Gaussian-to-Gaussian superoperators}.
\newblock {\em J. Math. Phys.}, 56(5):052202, May 2015.

\bibitem[EW07]{Eisert2007-rl}
J~Eisert and M~M Wolf.
\newblock {Gaussian Quantum Channels}.
\newblock In {\em {Quantum Information with Continuous Variables of Atoms and Light}}, pages 23--42. Published by Imperial College Press and distributed by World Scientific Publishing Co., February 2007.

\bibitem[Fan76]{Fannes1976-il}
M~Fannes.
\newblock {Quasi-free states and automorphisms of the CCR-algebra}.
\newblock {\em Commun. Math. Phys.}, 51(1):55--66, February 1976.

\bibitem[FV75]{Fannes1975-dh}
M~Fannes and A~Verbeure.
\newblock {Gauge transformations and normal states of the {CCR}}.
\newblock {\em J. Math. Phys.}, 16(10):2086--2088, October 1975.

\bibitem[FV76]{Fannes1976-yg}
M~Fannes and A~Verbeure.
\newblock {Erratum: Gauge transformations and normal states of the {CCR}}.
\newblock {\em J. Math. Phys.}, 17(2):284--284, February 1976.

\bibitem[Gos06]{Gosson2006-lr}
Maurice~A Gosson.
\newblock {\em {Symplectic Geometry and Quantum Mechanics}}.
\newblock Advances in Partial Differential Equations. Birkhauser Verlag AG, Basel, Switzerland, 2006 edition, December 2006.

\bibitem[Guh08]{Guha2008-iw}
Saikat Guha.
\newblock {\em {{Multiple-User} Quantum Information Theory for Optical Communication Channels}}.
\newblock Ph.d. thesis, Massachusetts Institute of Technology Cambridge, May 2008.

\bibitem[HJ12]{Horn2012-xx}
Roger~A Horn and Charles~R Johnson.
\newblock {\em {Matrix Analysis}}.
\newblock Cambridge University Press, 2 edition, 2012.

\bibitem[Hol19]{Holevo2019-zo}
Alexander~S Holevo.
\newblock {\em {Quantum Systems, Channels, Information A Mathematical Introduction}}.
\newblock Texts and Monographs in Theoretical Physics. De Gruyter, 2019.

\bibitem[HW01]{Holevo2001-zc}
A~S Holevo and R~F Werner.
\newblock {Evaluating capacities of bosonic Gaussian channels}.
\newblock {\em Phys. Rev. A}, 63(3):032312, February 2001.

\bibitem[Joh18]{john2018infinite}
Tiju~Cherian John.
\newblock {\em Infinite mode quantum gaussian states}.
\newblock Ph.d. thesis, Indian Statistical Institute-Kolkata, 2018.

\bibitem[JP21]{TiKR21}
Tiju~Cherian John and K.~R. Parthasarathy.
\newblock {A common parametrization for finite mode Gaussian states, their symmetries, and associated contractions with some applications}.
\newblock {\em Journal of Mathematical Physics}, 62(2):022102, 02 2021.

\bibitem[NC10]{Nielsen_Chuang_2010}
Michael~A. Nielsen and Isaac~L. Chuang.
\newblock {\em Quantum Computation and Quantum Information: 10th Anniversary Edition}.
\newblock Cambridge University Press, 2010.

\bibitem[Par92]{Krp92}
K.~R. Parthasarathy.
\newblock {\em An introduction to quantum stochastic calculus}.
\newblock Modern Birkh\"{a}user Classics. Birkh\"{a}user/Springer Basel AG, Basel, 1992.
\newblock [2012 reprint of the 1992 original] [MR1164866].

\bibitem[Par10]{Parthasarathy2010-su}
K~R Parthasarathy.
\newblock {What is a {G}aussian state?}
\newblock {\em Commun. Stoch. Anal.}, 4(2):143--160, 2010.

\bibitem[Par13]{KRP12}
K.~R. Parthasarathy.
\newblock The symmetry group of gaussian states in ${L}^2(\mathbb{R}^{n})$, 2013.

\bibitem[Par15]{KRP15}
K.~R. Parthasarathy.
\newblock Symplectic dilations, {G}aussian states and {G}aussian channels.
\newblock {\em Indian J. Pure Appl. Math.}, 46(4):419--439, 2015.

\bibitem[Par22]{Parthasarathy_2022}
K.~R. Parthasarathy.
\newblock Twisted convolution quantum information channels, one-parameter semigroups and their generators.
\newblock {\em Infinite Dimensional Analysis, Quantum Probability and Related Topics}, 25(04), December 2022.

\bibitem[Pau02]{Pau02}
Vern Paulsen.
\newblock {\em Completely bounded maps and operator algebras}, volume~78 of {\em Cambridge Studies in Advanced Mathematics}.
\newblock Cambridge University Press, Cambridge, 2002.

\bibitem[Pol22a]{Pol22}
Damiano Poletti.
\newblock Characterization of gaussian quantum markov semigroups.
\newblock {\em Infinite Dimensional Analysis, Quantum Probability and Related Topics}, 25(03):2250014, 2022.

\bibitem[Pol22b]{Poletti2022-wk}
Damiano Poletti.
\newblock {Characterization of Gaussian Quantum Markov Semigroups}.
\newblock In {\em {Infinite Dimensional Analysis, Quantum Probability and Applications}}, pages 197--211. Springer International Publishing, 2022.

\bibitem[RZBB94]{Reck1994-un}
M~Reck, A~Zeilinger, H~J Bernstein, and P~Bertani.
\newblock {Experimental realization of any discrete unitary operator}.
\newblock {\em Phys. Rev. Lett.}, 73(1):58--61, July 1994.

\bibitem[Ser17]{Serafini2017-mz}
Alessio Serafini.
\newblock {\em {Quantum {C}ontinuous {V}ariables: {A} {P}rimer of {T}heoretical {M}ethods}}.
\newblock Apple Academic Press, Oakville, MO, June 2017.

\bibitem[Sti55]{Sti55}
W.~Forrest Stinespring.
\newblock Positive functions on {$C^*$}-algebras.
\newblock {\em Proc. Amer. Math. Soc.}, 6:211--216, 1955.

\bibitem[SW17]{Sabapathy2017-xp}
Krishna~Kumar Sabapathy and Andreas Winter.
\newblock {Non-Gaussian operations on bosonic modes of light: Photon-added Gaussian channels}.
\newblock {\em Phys. Rev. A}, 95(6):062309, June 2017.

\bibitem[Wat18]{Watrous_2018}
John Watrous.
\newblock {\em The Theory of Quantum Information}.
\newblock Cambridge University Press, 2018.

\bibitem[Wil13]{Wilde_2013}
Mark~M. Wilde.
\newblock {\em Quantum Information Theory}.
\newblock Cambridge University Press, 2013.

\bibitem[WPGP{\etalchar{+}}12]{Weedbrook2012-zl}
Christian Weedbrook, Stefano Pirandola, Raúl García-Patrón, Nicolas~J Cerf, Timothy~C Ralph, Jeffrey~H Shapiro, and Seth Lloyd.
\newblock {Gaussian quantum information}.
\newblock {\em Rev. Mod. Phys.}, 84(2):621--669, May 2012.

\end{thebibliography}

\end{document}